\documentclass[final,3p,times]{elsarticle}
\biboptions{sort&compress}
\usepackage{yfonts}
\usepackage{graphicx}% Include figure files
\usepackage{dcolumn}% Align table columns on decimal point
\usepackage{bm}% bold math
\usepackage{nicefrac}
\usepackage{amsmath,amsfonts}
\usepackage{romannum}
\usepackage{color}
\usepackage[dvipsnames]{xcolor}
\usepackage{slashed}
\usepackage{cancel}
\definecolor{darkblue}{RGB}{0,0,196}
\definecolor{darkgreen}{RGB}{0,120,0}
\usepackage[colorlinks=true,linktocpage=true,linkcolor=darkblue,citecolor=red,urlcolor=darkblue]{hyperref}
\usepackage{amssymb}
\usepackage{amsthm}
\usepackage{amsmath}
\usepackage{lipsum}
\usepackage[inkscapearea=page]{svg}
\usepackage{appendix}

\newcommand{\bea}{\begin{eqnarray}}
\newcommand{\eea}{\end{eqnarray}}
\newcommand{\bel}[1]{\begin{eqnarray}\label{#1}}
\newcommand{\eel}{\end{eqnarray}}

\newcommand{\CITn}[1]{\citep{#1}} 

\newcommand{\p}{\partial}
\newcommand{\dd}{\mathrm{d}}

\newcommand{\bv}{{\boldsymbol b}} 
 
\newcommand{\ev}{{\boldsymbol e}}

\newcommand{\pv}{{\boldsymbol p}}

\newcommand{\xv}{{\boldsymbol x}}

\newcommand{\f}[2]{\frac{#1}{#2}}
\newcommand{\onehalf}{{\nicefrac{1}{2}}} 
 
\newcommand{\threefourths}{{\nicefrac{3}{4}}} 

\def\spin{\,\textgoth{s:}}

\def\al{\alpha}

\def\be{\beta}

\def\ga{\gamma}

\def\de{\delta}

\def\ep{\epsilon}

\def\la{\lambda}

\def\rh{\rho}

\def\ta{\tau}
\def\si{\sigma}

\def\om{\omega}

\begin{document}

\begin{frontmatter}
\title{Perfect spin hydrodynamics at all orders in spin polarization}
\author[first]{Zbigniew Drogosz}
\ead{zbigniew.drogosz@alumni.uj.edu.pl}
\affiliation[first]{organization={Institute of Theoretical Physics, Jagiellonian University},
            city={Kraków},
            postcode={30-348}, 
            country={Poland}}
\date{\today}
\begin{abstract}
We compare two recently developed frameworks of perfect spin hydrodynamics for spin-$1/2$ particles, based respectively on classical kinetic theory and the Wigner function. We show that the conserved currents in both approaches have the same form at each order of the expansion in the components of the spin polarization tensor $\omega$. The only difference is a relative multiplicative factor, which is equal to 1 at the lowest nontrivial order and increases monotonically with the expansion order.
\end{abstract}


\end{frontmatter}

\section{Introduction}

Since the birth of quantum mechanics, the relationship between classical and quantum physics has been widely discussed. The correspondence principle states that the classical picture is recovered by quantum physics in the limit of large quantum numbers, i.e., when the relevant actions are large compared to Planck's constant. Thus, for example, for large values of the angular momentum, the impact of its being quantized becomes negligible. This naturally raises the issue of the validity of a classical description of spin, whose values of the order of a few $\hbar$ can hardly be characterized as large quantum numbers. Nonetheless, in the field of relativistic hydrodynamics an interest in the classical spin exists, since polarization phenomena in heavy-ion collisions, in particular the polarization of $\Lambda$ hyperons~\cite{STAR:2017ckg, STAR:2018gyt, STAR:2019erd} and vector mesons~\cite{ALICE:2019aid}, have been described within the theoretical framework of spin hydrodynamics using kinetic theory and a classical spin picture introduced by Mathisson~\cite{Mathisson:1937zz,2010GReGr..42.1011M,Florkowski:2024bfw,Drogosz:2024gzv}.
That approach is far from being the only one possible, and progress towards the understanding of spin hydrodynamics has been made through several pathways,
including works where the final particle polarization is determined using gradients of hydrodynamic fields on the freezeout hypersurface~\cite{Becattini:2009wh,Becattini:2021iol,Palermo:2024tza},
kinetic theory studies~\cite{Florkowski:2017ruc,Shi:2020htn,Hu:2021pwh,  Bhadury:2020puc, Bhadury:2022ulr, Weickgenannt:2019dks, Weickgenannt:2021cuo, Weickgenannt:2020aaf, Weickgenannt:2022zxs, Weickgenannt:2023nge, Wagner:2024fhf, Banerjee:2024xnd, Bhadury:2024ckc}, 
derivations based on the entropy principle~\cite{Li:2020eon,
Hattori:2019lfp, Fukushima:2020ucl,
Biswas:2023qsw, Xie:2023gbo, Daher:2024ixz, Ren:2024pur, Daher:2024bah, 
Fang:2025aig}, Lagrangian effective theories~\cite{Montenegro:2017rbu,Montenegro:2020paq},
and a divergence-type theory framework~\cite{Abboud:2025shb}.
See also the reviews~\cite{Florkowski:2018fap,Becattini:2020ngo,Huang:2024ffg,Florkowski:2024cif}.

Recently, there appeared a new quantum spin hydrodynamics description~\cite{Bhadury:2025boe, Kar:2025qvj}, based on a novel Wigner function, aiming to supersede the previously used one~\cite{Becattini:2013fla}. It has several advantages: It leads to a well-defined polarization magnitude, thereby removing normalization issues; its application range covers parameter range that occurs in heavy-ion collisions~\cite{Drogosz:2025ihp}; it results in the same thermodynamic relations as those recently derived in a different way~\cite{Florkowski:2024bfw,Drogosz:2024gzv}; its equations have the desirable causality and stability properties~\cite{Bhadury:2025wuh}.
Somewhat surprisingly, the classical and the quantum perfect spin hydrodynamics pictures agree when spin polarizations are low~\cite{Bhadury:2025boe}, contrary to the usual requirement of large quantum numbers. This prompts the question whether and when this correspondence breaks down.

This study compares the quantum and the classical description of spin-polarized hydrodynamic systems and shows that the conserved current tensors in the two approaches have fully analogous structures at any order of expansion in the dimensionless spin polarization tensor $\omega_{\alpha \beta}$, defined as the ratio of the spin chemical potential $\Omega_{\alpha \beta}$ to the temperature $T$,
$\omega_{\alpha \beta} \equiv \Omega_{\alpha \beta}/T$. 
The expressions differ only by a multiplicative factor, which increases monotonically with the expansion order. This separation at higher orders can be precisely quantified. With the classical spin normalization constant $\spin$ equal to the eigenvalue of the SU(2) Casimir operator~\cite{Florkowski:2018fap}, $\spin^2= \threefourths$, chosen such that the baryon current~$N^\mu$ and the energy-momentum tensor $T^{\mu \nu}$ agree with the quantum description exactly at the second order in $\omega$, and the spin tensor $S^{\la, \mu \nu}$ agrees at the first order in $\omega$, the $2n$-th-order contribution to $N^\mu$ and $T^{\mu \nu}$ ($2n-1$-th-order contribution to $S^{\la, \mu \nu}$) is $3^n/(2n+1)$ times higher in the classical spin picture.

Notations and conventions: We use natural units, $\hbar = c = k_{\rm B} = 1$, the mostly negative metric convention $g^{\mu \nu} = \rm{diag}(1,-1,-1,-1)$, and the Levi-Civita symbol $\epsilon^{0123} = - \epsilon_{0123} = 1$. A colon between rank-two tensors denotes contraction over both indices, $\omega : s \equiv \omega_{\alpha \beta} s^{\alpha \beta}$. Round brackets around tensor indices denote normalized symmetrization, i.e., a sum over all permutations of the indices divided by the number of permutations, $T^{(\alpha_1 \alpha_2 \dots \alpha_n)} = \f{1}{n!}\sum\limits_{\si} T^{\si(\al_1) \si(\al_2) \dots \si(\al_n)}$. 

\section{Classical spin description}

The classical-spin approach to spin hydrodynamics \cite{Drogosz:2024gzv} uses particle (plus sign) and antiparticle (minus sign) equilibrium distribution functions $f_{\rm eq}^\pm(x,p,s)$ in the 
phase space that is extended to include the spin four-vector $s^\mu$ in addition to the 
spacetime coordinates $x^\mu = (t, \xv)$ and the four-momentum $p^\mu = (E_p, \pv)$. The mass-shell assumption gives $E_p =\sqrt{m^2 + \pv^2}$,
and the integration measure in the momentum space is
\begin{equation}
\dd P = \f{\dd^3 p}{(2\pi)^3 E_p}.
\end{equation}
%In the particle rest frame~(PRF) defined by $p^\mu = (m,0,0,0)$, $s^\alpha = (0,\sv_*)$, with the normalization~$\spinl = \spin$.
The internal angular momentum of a particle is defined in a way proposed by Mathisson~\CITn{Mathisson:1937zz,2010GReGr..42.1011M}, $ s^{\alpha \beta} = (1/m) \epsilon^{\alpha\beta\gamma\delta} p_\gamma s_\delta$. The spin four-vector $s^\mu = (1/(2m))\epsilon^{\al \be \ga \de} p_\be s_{\ga \de}$ satisfies the orthogonality and normalization conditions, $p \cdot s = 0$, $s^2 = -\spin^2$, which are enforced by delta functions in the Lorentz-invariant spin-space integration measure 
\begin{equation}
\dd S = \f{m}{\pi \spin} \dd^4 s \delta (s \cdot s + \spin^2) \delta (p \cdot s).
\end{equation}
In the main text, we assume a dilute system, in which it is justified to approximate the Fermi--Dirac particle ($+$) and antiparticle ($-$) distribution functions with the Boltzmann distribution
\begin{equation}
f^\pm_{\rm0, eq}(x,p,s) = \exp \bigg(\!\pm \xi-p \cdot \beta(x) + \frac{1}{2} \omega(x) : s \bigg),
\end{equation}
where $\xi \equiv \mu/T$ and $\beta^\mu \equiv u^\mu/T$, with $\mu$ being the chemical potential and $u^\mu$ the flow four-vector (see \ref{sec:fd} for an extension to the Fermi--Dirac case).
In the GLW pseudogauge \cite{DeGroot:1980dk}, chosen so that the spin tensor is nonzero and conserved in a perfect spin fluid, the scalar density
\begin{align}\label{ncl}
n_{\rm cl} =2 \int \dd P \int \dd S  \cosh \xi\exp\left(- p \cdot \beta + \frac{1}{2} \omega_{\alpha \beta} s^{\alpha \beta} \right) 
= 2\int \dd P   \cosh \xi\exp\left(- p \cdot \beta \right)
\int \dd S \exp\left(\frac{1}{2} \omega : s\right)
\end{align}
is the central object of the theory and a generating function for the equilibrium baryon current $N^\mu_{\rm eq}$, the energy-momentum tensor $T^{\mu \nu}_{\rm eq}$, and the spin tensor $S_{\rm eq} ^{\lambda, \mu\nu}$, which can be obtained from it by taking derivatives with respect to $\xi$, $\beta^\mu$ and $\omega_{\mu\nu}$,
\begin{align}\label{clN}
N^\mu_{\rm eq} &= -\f{\p^2 n_{\rm cl}}{\p \be_\mu \p \xi} =\int \dd P \,\dd S \, p^\mu \, \left[f_{\rm 0,eq}^+(x,p,s)-f_{\rm 0,eq}^-(x,p,s) \right],\\ \label{clT}
T^{\mu \nu}_{\rm eq} &= \f{\p^2 n_{\rm cl}}{\p \be_\mu \p\be_\nu} = \int \dd P \,\dd S \, p^\mu p^\nu \, \left[f_{\rm 0,eq}^+(x,p,s) + f_{\rm 0,eq}^-(x,p,s) \right],\\ \label{clS}
S_{\rm eq}^{\lambda, \mu\nu} &= -\f{\p^2 n_{\rm cl}}{\p \be_\lambda \p\omega_{\mu \nu}}  =\int \!\dd P \, \dd S \, \, p^\lambda \, s^{\mu \nu} 
\left[f_{\rm 0,eq}^+(x,p,s)+ f_{\rm 0,eq}^-(x,p,s) \right].
\end{align}

\section{Quantum spin description}

Assuming the Boltzmann distribution, the GLW versions of the conserved currents can be obtained by differentiating the quantum scalar density~\cite{Bhadury:2025boe,Drogosz:2025ihp}
\begin{equation}\label{nqt}
n_{\rm qt}(x) = 4  \int \! \dd P  \exp\left(- p \cdot \beta \right) \cosh \xi \,\cosh\!\sqrt{-a^2},
\end{equation}
namely
\begin{align}\label{qtN}
N^\mu(x) &= -\f{\p^2 n_{\rm qt}}{\p \be_\mu \p \xi} = 2  \int \dd P  \, p^\mu \left[  f_0^+(x,p)   - f_0^- (x,p)   \right]  \cosh\!\sqrt{-a^2},\\ \label{qtT}
T^{\mu\nu}(x) &= \f{\p^2 n_{\rm qt}}{\p \be_\mu \p\be_\nu}= 2  \int \dd P  \, p^\mu p^\nu \left[  f_0^+(x,p)   + f_0^- (x,p)   \right]  \cosh\!\sqrt{-a^2},\\ \label{qtS}
S^{\lambda, \mu\nu}(x) &= -\f{\p^2 n_{\rm qt}}{\p \be_\lambda \p \omega_{\mu \nu}} = \frac{1}{m} \int \dd P \, p^\lambda\, 
\left[ f_0^+(x,p) + f_0^-(x,p)  \right] \frac{\sinh\!\sqrt{-a^2}}{\sqrt{-a^2}}  {\epsilon^{\mu \nu}}_{\rho \sigma} a^\rho p^\sigma,
\end{align}
where
\begin{eqnarray}\label{amu}
a_\mu(x,p) \equiv -\frac{1}{2 m} {\tilde \omega}_{\mu\nu}(x)p^\nu = - \f{1}{4m} \epsilon_{\mu\nu\alpha \beta} \omega^{\alpha \beta} p^\nu.
\end{eqnarray}%
The last expression uses the definition of the dual spin polarization tensor ${\tilde \omega}_{\mu\nu} \equiv (1/2) \, \epsilon_{\mu\nu\alpha \beta} \omega^{\alpha \beta}$. If the spin polarization tensor $\omega_{\mu\nu}$ is expressed in terms of the electriclike and magneticlike three-vectors, \mbox{$\ev = (e^1,e^2,e^3)$} and $\bv = (b^1,b^2,b^3)$,
\begin{equation}\label{omegamatrix}
\omega_{\mu\nu} = 
\begin{bmatrix}
0     &  e^1 & e^2 & e^3 \\
-e^1  &  0    & -b^3 & b^2 \\
-e^2  &  b^3 & 0 & -b^1 \\
-e^3  & -b^2 & b^1 & 0
\end{bmatrix},
\end{equation}
then the form of the dual tensor ${\tilde \omega}_{\mu\nu}$ is obtained through the substitution $\ev \rightarrow \bv$ and $\bv \rightarrow -\ev$.

\section{Comparison of the quantum and classical generating functions}

Let us compare the generating functions (\ref{ncl}) and (\ref{nqt}). The quantum generating function can be expanded in the powers of $a^\mu$, which, as follows from the definition (\ref{amu}), is in fact an expansion in the powers of the components of $\om_{\mu \nu}$,
\begin{equation}\label{expqt}
n_{\rm qt}(x) = 4  \int \! \dd P  \exp\left(- p \cdot \beta \right) \cosh \xi \,\cosh\!\sqrt{-a^2} = 4  \int \! \dd P  \exp\left(- p \cdot \beta \right) \cosh \xi \sum_{n=0}^\infty (-1)^n \f{(-a^{2})^n }{(2n)!}.
\end{equation}
The exponential function of $\om : s$ that appears in the classical generating function can be expanded in a power series as well,
\begin{equation}\label{expcl}
n_{\rm cl} = 2 \int \dd P   \exp\left(- p \cdot \beta \right) \cosh \xi
\int \dd S \sum_{n=0}^\infty \f{(\omega : s)^n}{n!2^n}.
\end{equation}
Thus, the comparison requires the calculation of integrals over the spin space, $\int \dd S \ (\omega : s)^n$, $ n\geq 0$.
Quantities other than the spin four-vectors can be factored out of the integral,
\begin{align}\begin{split}\label{oms2n}
\int \dd S \, (\omega : s)^{n} &= \f{1}{m^n}\ep_{\al_1 \be_1 \mu_1 \nu_1}\, \ep_{\al_2 \be_2 \mu_2 \nu_2} \cdots \ep_{\al_{n} \be_{n} \mu_{n} \nu_{n}} \ \om^{\mu_1 \nu_1} \om^{\mu_2 \nu_2} \cdots \om^{\mu_{n} \nu_{n}} \ p^{\be_1} p^{\be_2} \cdots p^{\be_{n}}
\int \dd S s^{\al_1} s^{\al_2} \cdots s^{\al_{n}}.
\end{split}\end{align}
The integral of a product of spin four-vectors has a closed form in terms of the metrics and the $p$ four-vector (see the proof in the Appendix). It vanishes for products of an odd number of spin four-vectors (\ref{eq:odd}), 
whereas for even-length products it is proportional to the symmetrized product of projectors (\ref{eq:even}),
\begin{equation}\label{ints}
\int \dd S s^{\alpha_1} s^{\alpha_2} \cdots s^{\alpha_{2n}} =  \f{2\spin^{2n}}{2n+1} \Delta^{(\alpha_1 \alpha_2} \Delta^{\alpha_3 \alpha_4} \cdots \Delta^{\alpha_{2n-1} \alpha_{2n})}, \quad \Delta^{\mu \nu} \equiv g^{\mu \nu} - \f{p^\mu p^\nu}{m^2}.
\end{equation}

After the substitution of (\ref{ints}) into (\ref{oms2n}) and expansion of all projectors $\Delta$ using their definition, only the terms where every $\Delta$ contributed a metric tensor factor survive: The expression (\ref{oms2n}) contains at least $n$ factors of $p$, so multiplication by any additional such factor causes a Levi-Civita symbol to be contracted with two four-vectors $p$ and therefore vanish as a contraction of an antisymmetric tensor with a symmetric one. Hence, 
\begin{align}\begin{split}
\int \dd S \, (\omega : s)^{2n} &= \f{1}{m^{2n}}\ep_{\al_1 \be_1 \mu_1 \nu_1}\, \ep_{\al_2 \be_2 \mu_2 \nu_2} \cdots \ep_{\al_{2n} \be_{2n} \mu_{2n} \nu_{2n}} \ \om^{\mu_1 \nu_1} \om^{\mu_2 \nu_2} \cdots \om^{\mu_{2n} \nu_{2n}} \ p^{\be_1} p^{\be_2} \cdots p^{\be_{2n}}
\f{2\spin^{2n}}{2n+1} g^{(\al_1 \al_2} g^{\al_3 \al_4} \cdots g^{\al_{2n-1} \al_{2n})}.
\end{split}\end{align}
Because of the summation over every index, the contribution of each of the terms of the symmetrized product is equal. Therefore, the normalized symmetrized product can be replaced by the ordinary product $g^{\alpha_1 \alpha_2} g^{\alpha_3 \alpha_4} \cdots g^{\alpha_{2n-1} \alpha_{2n}}$,
\begin{align}\begin{split}\label{oms}
\int \dd S \, (\omega : s)^{2n} &= \f{1}{m^{2n}}\ep_{\al_1 \be_1 \mu_1 \nu_1}\, \ep_{\al_2 \be_2 \mu_2 \nu_2} \cdots \ep_{\al_{2n} \be_{2n} \mu_{2n} \nu_{2n}} \ \om^{\mu_1 \nu_1} \om^{\mu_2 \nu_2} \cdots \om^{\mu_{2n} \nu_{2n}} \ p^{\be_1} p^{\be_2} \cdots p^{\be_{2n}}
\f{2\spin^{2n}}{2n+1} g^{\alpha_1 \alpha_2} g^{\alpha_3 \alpha_4} \cdots g^{\alpha_{2n-1} \alpha_{2n}}\\
&= \f{2 \spin^{2n}}{2n+1} \left(\f{1}{m^2}\ep_{\al_1 \be_1 \mu_1 \nu_1}\, \ep^{\al_1}{}_{\be_2 \mu_2 \nu_2}\om^{\mu_1 \nu_1} \om^{\mu_2 \nu_2}p^{\be_1} p^{\be_2}\right)^n =(-1)^n \frac{2\spin^{2n} (-a^2)^n 2^{4n}}{2n+1},
\end{split}\end{align}
and thus
\begin{equation}\label{expcls}
n_{\rm cl} = 4 \int \dd P \exp(-p \cdot \beta) \cosh \xi \sum_{n=0}^\infty (-1)^n \f{ \spin^{2n} (-a^2)^n 2^{2n}}{(2n+1)!}.
\end{equation}

This result shows that the corresponding terms of the expansions of the quantum scalar density (\ref{expqt}) and the classical scalar density (\ref{expcls}) have the same functional form. However, at each order of the expansion they differ by a relative multiplicative factor. To achieve an exact equivalence, the classical spin normalization (which is the only parameter available) would have to absorb that factor and therefore no longer be a constant but vary with the expansion order, namely
\begin{equation}
\spin_{2n} = \frac{(2n+1)^{\f{1}{2n}}}{2}
\end{equation}
at the order $2n$ in the components of $\omega$.

In particular, for the second order, $\spin_2 = \sqrt{\threefourths}$, which coincides with the value of $\spin$ considered in previous works (Refs.~\cite{Florkowski:2018fap,Drogosz:2024gzv,Bhadury:2025boe}), whereas in the limit of infinite order
\begin{equation}
\lim_{n \rightarrow \infty} \spin_n = \f12,
\end{equation}
which explains the correspondence between application ranges discovered in Ref.~\cite{Drogosz:2025ihp}.

If one sets $\spin = \sqrt{\threefourths}$, the classical and quantum spin pictures exactly agree only at the lowest nontrivial order and exponentially diverge from each other at higher orders. The $2n$-th order contribution to the generating function is then
\begin{equation}
\bigg(\f{\spin}{\spin_{2n}}\bigg)^{2n}= \f{3^n}{2n+1}
\end{equation}
times larger in the classical description.

Only even orders in $\omega$ are present in the expansions of the baryon current and the energy-momentum tensor, and only odd orders are present in the expansion of the spin current. The integration over the momentum space in Eqs.~(\ref{clN})–(\ref{clS}) was performed explicitly at the lowest nontrivial order in $\omega$ in \cite{Drogosz:2024gzv} (and then in \cite{Drogosz:2025ose} without the Boltzmann approximation). In principle, it may be performed for any higher order, although the complexity increases rapidly.

\section{Summary}

This work proved that in the context of perfect spin hydrodynamics of spin-$\onehalf$ particles the classical and quantum spin descriptions lead to the same form of the conserved current tensors at each order in the expansion of the components of the spin polarization tensor $\omega_{\mu \nu}$. At the lowest nontrivial order, the expressions are exactly the same, whereas at higher orders they 
separate and differ by a relative multiplicative factor that increases exponentially with the expansion order. This shows that the classical description is valid for sufficiently low values of the components of $\omega_{\mu \nu}$, where the higher-order contributions are negligible.
Conceptually, this correspondence is possible because classically we consider a system of many particles and the  polarization averaging out to a low value.

The correspondence between the expressions explains also the similarity in application range formulas derived in Ref.~\cite{Drogosz:2025ihp}. In that derivation, the behavior of integrands that appear in the generating functions was analyzed in the regime of large momenta and polarizations. Therefore, it was the infinite-order limit that was captured. One could very informally say that the quantum approach results are obtained by reducing the classical spin normalization $\spin = \sqrt{3}/2$, valid for low polarizations, gradually to the standard spin value of $1/2$ at the limit of large polarizations.

From a practical point of view, especially valuable would be an implementation of any of these approaches to spin hydrodynamics
in 3D numerical simulations, similar to those performed recently using different dynamic equations 
\cite{Singh:2024cub, Sapna:2025yss}.

\section*{Acknowledgments} The author thanks Wojciech Florkowski for the inspiring discussions. This research was supported in part by National Science Centre, Poland (NCN) (Grant No. 2022/47/B/ST2/01372).

\appendix

\section{Calculation of an integral of a product of spin four-vectors}

The rank-$n$ tensor $I^{\alpha_1 \alpha_2 \dots \alpha_{n}}$ defined as
\begin{equation}\label{eq:I}
I^{\alpha_1 \alpha_2 \dots \alpha_{n}} \equiv \int \dd S s^{\alpha_1} s^{\alpha_2} \cdots s^{\alpha_{n}}
\end{equation}
is fully symmetric, and it is orthogonal to $p$ in every index, as can be verified through a direct calculation
\begin{equation}\label{contrp}
I^{\alpha_1 \alpha_2 \dots \alpha_{n}} p_{\alpha_1} = \int \dd S s^{\alpha_1} s^{\alpha_2} \cdots s^{\alpha_{n}} p_{\alpha_1} = \f{m}{\pi \spin} \int \dd^4 \!s \delta (s \cdot s + \spin^2) \delta (p \cdot s) (p \cdot s) s^{\alpha_2} \cdots s^{\alpha_{n}} = 0,
\end{equation}
where the factor $\delta (p \cdot s) (p \cdot s)$ causes the integral to vanish. Orthogonality to $p$ in every index entails invariance under contraction of any indices with the projector operator $\Delta_{\mu \nu} \equiv g_{\mu \nu} - \f{p_\mu p_\nu}{m^2}$ (the particle mass $m$ in the definition ensuring that the projector is dimensionless), 
\begin{equation}\label{eq:inv}
\Delta^{\beta}{}_{\alpha_i} I^{\alpha_1 \alpha_2 \dots \alpha_{i-1}\alpha_i\alpha_{i+1}\dots\alpha_{n}} = I^{\alpha_1 \alpha_2 \dots \alpha_{i-1}\beta\alpha_{i+1}\dots\alpha_{n}}.
\end{equation}

\newtheorem{lemma}{Lemma}
\begin{lemma}
$I^{\alpha_1 \alpha_2 \dots \alpha_{n}}$ is a polynomial in $\Delta^{\mu \nu}$.
\end{lemma}
\begin{proof}
The only available objects from which $I^{\alpha_1 \alpha_2 \dots \alpha_{n}}$ has to be constructed are the metric tensor $g^{\mu \nu}$ and the four-momentum $p^\mu$, and the form of the right-hand side of (\ref{eq:I}) implies that $I^{\alpha_1 \alpha_2 \dots \alpha_{n}}$ is a polynomial in those quantities. 

Consider any monomial in $\Delta$, $g$ and $p$ that is a part of $I^{\alpha_1 \alpha_2 \dots \alpha_{n}}$.
A monomial built solely of projectors $\Delta$ is invariant under the action of any projector $\Delta$.
A monomial that contains a free $p^\mu$ (i.e., one that is not included in any projector) is annihilated by $\Delta^\nu{}_\mu$. A free $g$ within a monomial turns into $\Delta$ when acted upon by a suitable projector.

Therefore, for the invariance requirement (\ref{eq:inv}) to be satisfied, either all monomials must be built only of $\Delta$ or the annihilation of any monomial $Q_1$ with free $p$ by a projector must be compensated by an appearance of terms with $p$ when a matching monomial $Q_2$ with free $g$ is acted upon by the same projector. However, for the coefficients to agree in the latter case, $Q_1$ and $Q_2$ must contain a common factor, and their sum must apart from this factor contain a linear combination that can be written as a $\Delta$,
\begin{equation}
Q_1 + Q_2 = Q(\Delta, g,p)\bigg(g^{\mu \nu} - \f{p^{\mu} p^{\nu}}{m^2} \bigg) = Q\Delta^{\mu \nu} \ \ {\rm for \ some} \ \mu, \nu.
\end{equation}
Repeating this coefficient-matching argument for every index shows that if (\ref{eq:inv}) holds, then the sum of monomials with free $p$ and monomials with free $g$ can in fact been written as monomials in $\Delta$ only, and, therefore, the entire $I$ is a polynomial in $\Delta$.
\end{proof}

From Lemma 1 it follows that all odd-rank tensors $I$ vanish,
\begin{equation}\label{eq:odd}
I^{\alpha_1 \alpha_2 \dots \alpha_{2n+1}} = \int \dd S s^{\alpha_1} s^{\alpha_2} \cdots s^{\alpha_{2n+1}} = 0,
\end{equation}
as it is impossible to construct an odd-rank tensor only from rank-two tensors $\Delta^{\mu \nu}$.

\begin{lemma}
Up to normalization, the even-rank tensors $I$ are symmetrized sums of products of $\Delta$,
\begin{equation}\label{idelta}
I^{\alpha_1 \alpha_2 \dots \alpha_{2n}} = \int \dd S s^{\alpha_1} s^{\alpha_2} \cdots s^{\alpha_{2n}} =  C_n \Delta^{(\alpha_1 \alpha_2} \Delta^{\alpha_3 \alpha_4} \cdots \Delta^{\alpha_{2n-1} \alpha_{2n})}.
\end{equation}
\end{lemma}
\begin{proof}
Since $I$ is a rank-$2n$ tensor constructed of $\Delta$, its most general form is a linear combination of products of $\Delta$, 
\begin{equation}\label{symm1}
I^{\al_1 \al_2 \dots \al_{2n}} = \sum_i c_i \Delta^{\alpha_{\sigma_i(1)} \alpha_{\sigma_i(2)}}\Delta^{\alpha_{\sigma_i(3)} \alpha_{\si_i(4)}} \cdots \Delta^{\alpha_{\si_i(2n-1)} \alpha_{\si_i(2n)}},
\end{equation}
where $i$ enumerates inequivalent $2n$-element permutations $\sigma_i$ (which are in a one-to-one correspondence with the $(2n-1)!!$ partitions of a $2n$-element set into $n$ two-element sets). Now, since $I$ is fully symmetric, it is invariant under any permutation $\rho$ of its indices,
\begin{align}\begin{split}\label{symm2}
I^{\al_1 \al_2 \dots \al_{2n}} &=I^{\al_{\rho(1)} \al_{\rho(2)} \dots \al_{\rho(2n)}} = \sum_i c_i \Delta^{\alpha_{\sigma_i \circ \rho(1)} \alpha_{\sigma_i \circ \rho(2)}}\Delta^{\alpha_{\sigma_i \circ \rho(3)} \alpha_{\sigma_i \circ \rho(4)}} \cdots \Delta^{\alpha_{\sigma_i \circ \rho(2n-1)} \alpha_{\sigma_i \circ \rho(2n)}} \\
&= \sum_i c_{\rho^{-1} (i)} \Delta^{\alpha_{\sigma_i(1)} \alpha_{\sigma_i(2)}}\Delta^{\alpha_{\sigma_i(3)} \alpha_{\si_i(4)}} \cdots \Delta^{\alpha_{\si_i(2n-1)} \alpha_{\si_i(2n)}}.
\end{split}\end{align}
Comparing (\ref{symm1}) and (\ref{symm2}) and using the linear independence of the products of projectors leads to the equality of linear combination coefficients
\begin{equation}
c_{\rho(i)} = c_i
\end{equation}
for any permutation $\rho$. Therefore, all the $c_i$ are equal, which proves (\ref{idelta}).
\end{proof}

The normalization coefficient $C_n$ can be calculated recursively. Let us calculate a double contraction of $\Delta$ and $I$,

\begin{align}\begin{split}
J^{\al_1 \al_2 \dots \al_{2n-2}} &\equiv\Delta_{\alpha_{2n-1} \alpha_{2n}} I^{\alpha_1 \alpha_2 \dots \al_{2n-1} \alpha_{2n}} = C_n \Delta_{\alpha_{2n-1} \alpha_{2n}} \Delta^{(\alpha_1 \alpha_2} \Delta^{\alpha_3 \alpha_4} \cdots \Delta^{\alpha_{2n-1} \alpha_{2n})} \\
&= C_n \f{1}{(2n-1)!!}\Delta_{\alpha_{2n-1} \alpha_{2n}} \sum_i \Delta^{\alpha_{\sigma_i(1)} \alpha_{\sigma_i(2)}}\Delta^{\alpha_{\sigma_i(3)} \alpha_{\si_i(4)}} \cdots \Delta^{\alpha_{\si_i(2n-1)} \alpha_{\si_i(2n)}},
\end{split}\end{align}
where $i$ enumerates inequivalent permutations (compare the note below Eq.~(\ref{symm1})). For each of the terms of the sum, there are two possibilities. Either the indices $\alpha_{2n-1}$ and $\al_{2n}$ belong to the same projector operator or to two different ones. There are $(2n-3)!!$ terms that correspond to the former possibility (the number of partitions of a $2n-2$-element set into $n-1$ two-element sets), while the remaining $(2n-2)((2n-3)!!)$ terms correspond to the latter one. Thus,
\begin{align}\begin{split}
J^{\al_1 \al_2 \dots \al_{2n-2}}  &= C_n \f{1}{(2n-1)!!} \bigg(\Delta_{\alpha_{2n-1} \alpha_{2n}} \Delta^{\alpha_{2n-1} \alpha_{2n}}\sum_i \prod_{m=1}^{n-1}\Delta^{\alpha_{\rh_i(2m-1)} \alpha_{\rh_i(2m)}}\\
&+\Delta_{\alpha_{2n-1} \alpha_{2n}} \sum_{j,k}\Delta^{\alpha_{2n-1} \alpha_j}\Delta^{\alpha_{2n-1} \alpha_k} \sum_i \prod_{m=1}^{n-2} \Delta^{\alpha_{\ta_i(f(2m-1))} \alpha_{\ta_i(f(2m))}} \bigg),
\end{split}\end{align}
where $i$ enumerates all the $(2n-3)!!$ inequivalent $2n-2$-element permutations $\rho_i$ in the first sum and all the $(2n-5)!!$ inequivalent $2n-4$-element permutations $\tau_i$ in the second sum, whereas $f$ is an arbitrary bijection $ f: \{1,...,2n-4\} \rightarrow \{1,...,2n-2\} \setminus \{j,k\}$. Performing the projection operator contractions,
\begin{align}
\Delta_{\alpha_{2n-1} \alpha_{2n}} \Delta^{\alpha_{2n-1} \alpha_{2n}} &=3,\\
\Delta_{\alpha_{2n-1} \alpha_{2n}} \Delta^{\alpha_{2n-1} \alpha_j}\Delta^{\alpha_{2n-1} \alpha_k} &= \Delta^{\al_j \al_k},
\end{align}
results in
\begin{align}\begin{split}\label{rec1}
J^{\al_1 \al_2 \dots \al_{2n-2}}   &= C_n \f{1}{(2n-1)!!} \bigg(3 \sum_i \prod_{m=1}^{n-1}\Delta^{\alpha_{\rh_i(2m-1)} \alpha_{\rh_i(2m)}} + \sum_{j,k}\Delta^{\al_j \al_k} \sum_i \prod_{m=1}^{n-2} \Delta^{\alpha_{\ta_i(f(2m-1))} \alpha_{\ta_i(f(2m))}} \bigg) \\
&=C_n \f{3 + 2n-2}{(2n-1)!!} \sum_i \prod_{m=1}^{n-1}\Delta^{\alpha_{\rh_i(2m-1)} \alpha_{\rh_i(2m)}} = C_n \f{2n+1}{2n-1}  \Delta^{(\alpha_1 \alpha_2} \Delta^{\alpha_3 \alpha_4} \cdots \Delta^{\alpha_{2n-3} \alpha_{2n-2})}.
\end{split}\end{align}
On the other hand, from the definition (\ref{eq:I}),
\begin{align}
\Delta_{\alpha_{2n-1} \alpha_{2n}} I^{\alpha_1 \alpha_2 \dots \al_{2n-1} \alpha_{2n}} = \int \dd S s^{\alpha_1} s^{\alpha_2} \cdots s^{\alpha_{2n}} \bigg(g_{\alpha_{2n-1} \alpha_{2n}} - \f{p_{\alpha_{2n-1}}p_{\alpha_{2n}}}{m^2} \bigg).
\end{align}
The term that contains the contractions with $p$ vanishes similarly as in (\ref{contrp}), whereas the contraction in the other term produces $\spin^2$, leaving
\begin{align}\label{rec2} 
J^{\al_1 \al_2 \dots \al_{2n-2}} =\spin^2 \int \dd S s^{\alpha_1} s^{\alpha_2} \cdots s^{\alpha_{2n-2}} = \spin^2 I^{\alpha_1 \alpha_2 \dots \al_{2n-3} \alpha_{2n-2}} = \spin^2 C_{n-1} \Delta^{(\alpha_1 \alpha_2} \Delta^{\alpha_3 \alpha_4} \cdots \Delta^{\alpha_{2n-3} \alpha_{2n-2})}.
\end{align}
Thus, from (\ref{rec1}) and (\ref{rec2}),
\begin{equation}
C_n = \spin^2 C_{n-1}  \f{2n-1}{2n+1}.
\end{equation}
Because $\int \dd S = 2 = C_0$ (normalization due to the constants in the definition of the measure $\dd S$\cite{Florkowski:2018fap}), the recursion is resolved as
\begin{equation}
C_n = \f{2 \spin^{2n}}{2n+1},
\end{equation}
leading to the result
\begin{equation}\label{eq:even}
\int \dd S s^{\alpha_1} s^{\alpha_2} \cdots s^{\alpha_{2n}} = \f{2\spin^{2n}}{2n+1} \Delta^{(\alpha_1 \alpha_2} \Delta^{\alpha_3 \alpha_4} \cdots \Delta^{\alpha_{2n-1} \alpha_{2n})}.
\end{equation}

\section{Generalization of the results to the Fermi--Dirac statistics}\label{sec:fd}

Following \cite{Kar:2025qvj} and \cite{Bhadury:2025wuh},
the generating functions for the case of spin-1/2 particles obeying the Fermi--Dirac distribution in the classical and the quantum approach are
\begin{equation}
n_{\rm clFD} = \int \dd P \int \dd S \big[F(y^+) + F(y^-)\big]
\end{equation}
and
\begin{equation}
n_{\rm qtFD} = \int \dd P \big[F(y^{++})+F(y^{+-}) + F(y^{-+})+F(y^{--}) \big],
\end{equation}
respectively, with the function $F$ defined as
\begin{equation}\label{deff}
F(y) = \f{y^2}{2} + {\rm Li}_2 (-e^y) +\f{\pi^2}{6},
\end{equation}
where ${\rm Li}_2$ is the dilogarithm. The symbol $y$ with the plus or minus indices is a shorthand notation for the exponent appearing in the distribution function, for particles or antiparticles (the first sign) with spin up or down (the second sign; the latter distinction appears only in the quantum framework),
\begin{align}
y^{\pm} &= \mp \xi + p \cdot \beta - \frac12 \om : s,\\
y^{\pm \pm} &= \mp \xi + p \cdot \beta \ \mp \sqrt{-a^2}.
\end{align}
It can be verified that the second derivative of the function $F$ has a familiar form,
\begin{equation}\label{secdev}
F''(y) = \frac{1}{1+\exp y} = \begin{cases}\sum_{k=1}^\infty(-1)^{k+1} \exp(-ky),\quad y>0,\\ \sum_{k=0}^\infty(-1)^k \exp(ky),\quad y < 0.\end{cases}
\end{equation}
For nondegenerate matter ($y > 0$), it is the first expansion that should be used.
Now, integration of the series term by term\footnote{One can verify by computing $F(y)$ at selected points and comparing it with the sum of the series that the constants of integration are $0$.} and application of (\ref{oms}) leads to Fermi--Dirac analogues of expressions (\ref{expqt}) and (\ref{expcls})
\begin{align}
n_{\rm cl FD} \! &= \!4 \!\! \int \!\!\dd P \sum_{k=1}^\infty \f{1}{k^2}\exp (-k p \!\cdot\! \beta) \cosh (k \xi) \cosh(k\!\sqrt{-a^2}) = \!\sum_{k=1}^\infty \!\f{4}{k^2} \!\! \int \! \! \dd P  \exp\left(- k p \!\cdot \!\beta \right) \cosh (k \xi) \!\sum_{n=0}^\infty (-1)^n \f{(-a^{2})^n k^{2n} }{(2n)!},\\
\begin{split}n_{\rm qt FD} &= 2 \int \dd P \int \dd S \sum_{k=1}^\infty \f{1}{k^2} \exp (-k p \cdot \beta) \cosh (k \xi)\exp \Big(\frac{k}{2} \om : s \Big)\\ &= \sum_{k=1}^\infty \f{4}{k^2} \int \dd P \exp(-k p \cdot \beta) \cosh (k\xi) \sum_{n=0}^\infty (-1)^n \f{ \spin^{2n} (-a^2)^n 2^{2n} k^{2n}}{(2n+1)!}.
\end{split}\end{align}
The first contribution, $k=1$, is the Boltzmann expression, as in the main text. For any higher $k$, the ratio of the term indexed by $(n,k)$ in the expansion of $n_{\rm cl FD}$ to the corresponding term in the expansion of $n_{\rm qt FD}$ is the same as that of the texms indexed by $(n,1)$. Each contribution is multiplied by the same factor in both the classical and the quantum expression, keeping the ratio between them the same as in the Boltzmann case.
Therefore, the relations between the classical and quantum framework derived within the Boltzmann approximation are valid for the Fermi--Dirac statistics as well.

\end{document}